\documentclass[preprint,12pt]{elsarticle}

\usepackage{breakurl}             
\usepackage{underscore}           
\usepackage{graphicx}
\usepackage{amssymb}
\usepackage{amsthm}
\usepackage{stmaryrd}
\usepackage{xcolor}
\usepackage{color,soul}
\usepackage{hyperref}
\usepackage{amsmath}
\usepackage{makecell}
\usepackage{fdsymbol}
\usepackage{tikz}
\usepackage{spverbatim}

\newcommand\scalemath[2]{\scalebox{#1}{\mbox{\ensuremath{\displaystyle #2}}}}

\newcommand{\tooleq}{${\it SDN\!\!-\!\!SafeCheck}$}

\newcommand{\svdash}{\vdash_s}
\newcommand{\PHbH}{{HbH}}
\newcommand{\inn}{{\it in}}
\newcommand{\out}{{\it out}}
\newcommand{\pt}{{\it pt}}

\newcommand{\dup}{{{\bf dup}}}

\newcommand{\NetKATdup}{{{\textnormal{NetKAT}}^{\textnormal{-\bf{dup},$*$}}}}

\newenvironment{todo}{\bigskip\hrule\medskip\noindent}{\medskip\hrule\bigskip}

\makeatletter
\newcommand{\pushright}[1]{\ifmeasuring@#1\else\omit\hfill$\displaystyle#1$\fi\ignorespaces}
\newcommand{\pushleft}[1]{\ifmeasuring@#1\else\omit$\displaystyle#1$\hfill\fi\ignorespaces}
\makeatother

\newtheorem{theorem}{Theorem}
\newtheorem{definition}{Definition}
\newtheorem{lemma}{Lemma}
\newtheorem{remark}{Remark}
\newtheorem{corollary}{Corollary}

\journal{JLAMP}
\date{}

\begin{document}

\begin{frontmatter}



\title{Explaining Safety Failures in NetKAT}


\author[KN]{Georgiana Caltais}
\ead{gcaltais@gmail.com}
\author[KN]{H\"unkar Can Tun\c c}
\ead{hcantunc@gmail.com}

\address[KN]{ Department for Computer and Information Science, University of Konstanz, Germany}

\begin{abstract}
This work introduces a concept of explanations with respect to the violation of safe behaviours within software defined networks (SDNs) expressible in NetKAT. The latter is a network programming language based on a well-studied mathematical structure, namely, Kleene Algebra with Tests (KAT). Amongst others, the mathematical foundation of NetKAT gave rise to a sound and complete equational theory. In our setting, a safe behaviour is characterised by a NetKAT policy, or program, which does not enable forwarding packets from an ingress $i$ to an undesirable egress $e$. We show how explanations for safety violations can be derived in an equational fashion, according to a modification of the existing NetKAT axiomatisation.
We propose an approach based on the Maude system for actually computing the undesired behaviours witnessing the forwarding of packets from $i$ to $e$ as above. 
{\tooleq{}} is a tool based on Maude equational theories satisfying important properties such as Church-Rosser and termination. {\tooleq{}} automatically identifies all the undesired behaviours leading to $e$, covering forwarding paths up to a user specified size. 
\end{abstract}




\end{frontmatter}


\section{Introduction}
\label{sec::intro}

Explaining systems failure has been a topic of interest for many years now.
Techniques such as Fault tree analysis (FTA) and Failure mode and effects analysis (FMEA)~\cite{Buckl:2007:MAD:1927558.1927572}, for instance, have been proposed and widely used by reliability engineers in order to understand how systems can fail, and for debugging purposes.

In this paper we focus on explaining violations of safe behaviours in software defined networks (SDNs).
{Software defined networking} is an emerging approach to network programming in a setting where the network control is decoupled from the forwarding functions. This makes the network control directly programmable, and more flexible to change.
SDN proposes open standards such as the OpenFlow~\cite{DBLP:journals/ccr/McKeownABPPRST08} API defining, for instance, low-level languages for handling switch configurations. Typically, this kind of hardware-oriented APIs are not intuitive to use in the development of programs for SDN platforms.
Hence, a suite of network programming languages raising the level of abstraction of programs, and corresponding verification tools have been recently proposed~\cite{DBLP:conf/icfp/FosterHFMRSW11,DBLP:conf/dsl/VoellmyH09,DBLP:conf/sigcomm/VoellmyWYFH13}.

It is a known fact that formal foundations can play an important role in guiding the development of programming languages and associated verification tools, in accordance with an intended semantics obeying essential (behavioural) laws.
Correspondingly, the current paper is targeting {NetKAT}~\cite{DBLP:conf/popl/AndersonFGJKSW14,DBLP:conf/popl/FosterKM0T15}  --a formal framework for specifying and reasoning about networks, integrated within the Frenetic suite of network management tools~\cite{DBLP:conf/icfp/FosterHFMRSW11}. In this work we exploit the sound and complete axiomatisation of NetKAT in~\cite{DBLP:conf/popl/AndersonFGJKSW14} and derive explanations of safety failures in a purely equational fashion.

From a more practical perspective, we introduce {\tooleq}, a tool based on the Maude system~\cite{DBLP:conf/rta/ClavelDELMMT03},
aiming at automatically computing the explanations for undesired behaviours within a NetKAT program that forwards packets from an ingress $i$ to an egress $e$. {\tooleq} is based on Maude confluent and terminating equational specifications, and computes the explanations for all the undesired behaviours covering forwarding paths up to a user specified size.

\medskip
Related to the current work,
the authors of NetKAT~\cite{DBLP:conf/popl/AndersonFGJKSW14} show that checking certain properties about networks, including reachability properties, can be reduced to equivalence checking problems in NetKAT by utilizing its sound and complete axiomatisation. NetKAT is also equipped with a practical tool which can check the equivalence of NetKAT policies~\cite{DBLP:conf/popl/FosterKM0T15}. The main focus of the tool proposed in~\cite{DBLP:conf/popl/FosterKM0T15} is to check whether a property holds in the network. This differs from our focus that we aim on discovering all possible ways a reachability property can be violated, and provide explanations that may be instructive for debugging purposes. 
\color{black}

The results in~\cite{root-fail-netkat} introduce a framework for automated failure localisation in NetKAT.
The approach in~\cite{root-fail-netkat} relies on the generation of test cases based on the network specification, further used to monitor the network traffic accordingly and localise faults whenever tests are not satisfied.
In contrast, our approach provides explanations for possible failures irrespective of particular input packets.

The work in~\cite{DBLP:conf/icfem/DengZL17} was the first to utilize a rewrite engine to manipulate NetKAT expressions in order to verify network properties. The authors of~\cite{DBLP:conf/icfem/DengZL17} propose an operational semantics for NetKAT and implement their formal specification in Maude. By utilizing the proposed operational semantics, the authors mainly follow three different techniques for automated reasoning in NetKAT: model checking of invariants, linear temporal logic based model checking, and normalization. The proposed formulations of the model checking procedures do not provide an explicit counterexample in case of a failure, hence these methods are unsuitable in our context. The normalization method is a different formulation of the equivalence checking approach that was proposed in~\cite{DBLP:conf/popl/AndersonFGJKSW14} for verifying network properties. The normalization method assesses whether NetKAT policies can be converted into the same normal form. This is a relevant method in our setting as well, however, the experimental evaluation in~\cite{DBLP:conf/icfem/DengZL17} shows that the proposed specification for the normalization approach fails to scale even for networks of moderate size. 
\color{black}


\medskip
{\it Our contributions.} This paper is an extension of our previous work in~\cite{DBLP:journals/corr/abs-1909-01745}. In~\cite{DBLP:journals/corr/abs-1909-01745} we introduced a concept of \emph{safety in NetKAT} which, in short, refers to the impossibility of packets to travel from a given ingress to a specified hazardous egress, in the context of the so-called ``port-based hop-by-hop" switch policies allowing only tests and port modifications. Then, we proposed a notion of \emph{safety failure explanation} which, intuitively, represents the set of finite paths within the network, leading to the hazardous egress. Eventually, we provided a modified version of the original axiomatisation of NetKAT exploited in order to \emph{automatically compute the safety failure explanations}, if any. The axiomatisation employed a proposed \emph{star-elimination construction} which enabled the sound extraction of explanations from Kleene $*$-free NetKAT programs. 

The current revised version of the paper extends the results in~\cite{DBLP:journals/corr/abs-1909-01745} as follows.
\begin{enumerate}
\item We propose a notion of safety in the context of more general switch policies defined as arbitrary expressions over the *-free, $\dup$-free fragment of NetKAT.
\item We show that a NetKAT network behaviour is ``safe'' whenever it can be proven so according to the proposed equational system used to derive safety failure explanations (see Corollary~\ref{cor::safe-compl}).
\item We formalize a concept of minimal, or relevant explanations for safety failures in NetKAT, based on a notion of ``normal forms for safety'' (see Section~\ref{sec:minimality}).
\item We introduce {\tooleq}, a practical tool for automatically computing safety failure explanations (see Section~{\ref{sec::tool}}). To the best of our knowledge, this tool is the first to provide automated failure explanations in NetKAT. 
\item We provide experimental evaluations for {\tooleq} based on the Topology Zoo dataset~\cite{DBLP:conf/pam/GillALM08}.
\end{enumerate}

\bigskip
\noindent
{\it Structure of the paper.} In Section~\ref{sec::prelim} we provide an overview of NetKAT and the associated sound and complete axiomatisation. In Section~\ref{sec::safety} we define the concept of safety in NetKAT and we introduce the notion of (minimal) safety failure explanation and the axiomatisation which can be exploited in order to compute such explanations.
In Section~\ref{sec::tool} we introduce the Maude-based tool {\tooleq}.
Experimental evaluation is discussed in Section~\ref{sec:experiments}.
In Section~\ref{sec::discussion} we draw the conclusions and pointers to future work.

\section{Preliminaries}
\label{sec::prelim}


As pointed out in~\cite{DBLP:conf/popl/AndersonFGJKSW14},
a network can be interpreted as an
automaton that forwards packets from one node to another along the links
in its topology. This lead to the idea of using regular expressions --the language of finite automata--, for expressing networks. A path is encoded as a concatenation of processing
steps ($p\cdot q.\ldots$), a set of paths is encoded as a union of paths
($p+q+\ldots$) whereas iterated processing is encoded using Kleene~$*$.
This paves the way to reasoning about properties of networks using Kleene Algebra with Tests (KAT)~\cite{DBLP:journals/toplas/Kozen97}. KAT incorporates both Kleene Algebra~\cite{DBLP:journals/iandc/Kozen94} for
reasoning about network structure and Boolean Algebra for reasoning about the predicates that define switch behaviour.


\begin{figure}[t]
\begin{minipage}[t]{0.45\linewidth}
{\scriptsize
\begin{alignat}{6}
\text{Fields}\quad&  &f::&{=}\enskip f_1\;| \dots |\;f_k \notag \\
\text{Packets}\quad& &pk::&{=}\enskip \{ f_1 = v_1, \dots, & \pushleft{f_k = v_k\}} & \notag \\
\text{Histories}\quad& &h::&{=}\enskip pk{:}{:} \langle \rangle \;| \; pk{:}{:}h & \notag\\
\text{Predicates}\quad& &a,b::&{=}\enskip 1 & \pushleft{Identity} & \notag\\
& & & | \quad 0 \quad & \pushleft{Drop}&\notag \\
& & & | \quad f=n \quad & \pushleft{Test}& \notag \\
& & & | \quad a+b \quad & \pushleft{Disjunction}& \notag \\
& & & | \quad a\cdot b \quad & \pushleft{Conjunction}& \notag \\
& & & | \quad \neg a \quad & \pushleft{Negation}& \notag \\ 
\text{Policies}\quad& &p,q::&{=}\enskip a & \pushleft{Filter} & \notag \\
& & & | \quad f\leftarrow n \quad & \pushleft{Modification}&\notag \\
& & & | \quad p+q \quad & \pushleft{Union}& \notag \\
& & & | \quad p\cdot q \quad & \pushleft{Sequential\;composition}& \notag \\
& & & | \quad p^* \quad & \pushleft{Kleene\;star}& \notag \\
& & & | \quad \dup \quad & \pushleft{Duplication}& \notag
\end{alignat}
}
\end{minipage}
\hfill
\begin{minipage}[t]{0.45\linewidth}
{\scriptsize
 \begin{alignat}{4}
\llbracket p \rrbracket \in &\;\text{H} \rightarrow \mathcal{P}(\text{H}) \notag \\
\llbracket 1 \rrbracket \; h \triangleq &\; \{ h \} \notag \\
\llbracket 0 \rrbracket \; h \triangleq &\; \{ \} \notag \\
\llbracket f=n \rrbracket \; (pk{:}{:}h) \triangleq &\; \left\{
                \begin{array}{ll}
                  \{pk{:}{:}h\} & \textnormal{if}\; pk.f = n \\
                  \{ \} & \textnormal{otherwise}
                \end{array}
              \right. \notag \\
\llbracket \neg a \rrbracket \; h \triangleq &\; \{h\} \setminus (\llbracket a \rrbracket \; h) \notag \\
\llbracket f \leftarrow n \rrbracket \; (pk{:}{:}h) \triangleq &\; \{ pk[f := n]{:}{:}h\} \notag \\
\llbracket p + q \rrbracket \; h \triangleq &\; \llbracket p \rrbracket \; h \cup \llbracket q \rrbracket \; h \notag\\
\llbracket p \cdot q \rrbracket \; h \triangleq &\; ( \llbracket p \rrbracket \bullet \llbracket q \rrbracket) \; h \notag \\
\llbracket p^* \rrbracket \; h \triangleq &\; \bigcup_{i \in N} F^i\;h \notag \\
F^0\; h \triangleq \{ h \} &\text{ and } F^{i+1}\;h \triangleq (\llbracket p \rrbracket \bullet F^i)\;h \notag \\
\llbracket \dup \rrbracket \; (pk{:}{:}h) \triangleq &\; \{pk{:}{:}(pk{:}{:}h)\} \notag
\end{alignat}
}
\end{minipage}
\caption{NetKAT Syntax and Semantics~\cite{DBLP:conf/popl/AndersonFGJKSW14}}
\label{fig::netKAT-syn-sem}
\end{figure}


\begin{figure}[ht]
\center
\[
\small{\scalemath{0.75}{
\begin{array}{r@{}c@{}ll|r@{}c@{}ll}
p + (q+r) & \,\equiv\, & (p+q) + r & \textnormal{\footnotesize KA-PLUS-ASSOC}~ &
~a + (b\cdot c) & \,\equiv\, &\, (a+b)\cdot(a+c) & \textnormal{\footnotesize BA-PLUS-DIST}\\
p + q& \,\equiv\, & q+p & \textnormal{\footnotesize KA-PLUS-COMM} &
a +1 & \,\equiv\, &\, 1 & \textnormal{\footnotesize BA-PLUS-ONE}\\
p + 0& \,\equiv\, &p & \textnormal{\footnotesize KA-PLUS-ZERO} &
a +\neg a & \,\equiv\, &\, 1 & \textnormal{\footnotesize BA-EXCL-MID}\\
p + p & \,\equiv\, & p & \textnormal{\footnotesize KA-PLUS-IDEM} &
a \cdot b & \,\equiv\, &\, b\cdot a & \textnormal{\footnotesize BA-SEQ-COMM}\\
p \cdot (q\cdot r) & \,\equiv\, & (p\cdot q) \cdot r & \textnormal{\footnotesize KA-SEQ-ASSOC}~ &
a \cdot \neg a & \,\equiv\, &\, 0 & \textnormal{\footnotesize BA-CONTRA}\\
1\cdot p & \,\equiv\, & p & \textnormal{\footnotesize KA-ONE-SEQ}~ &
a \cdot a & \,\equiv\, &\,a & \textnormal{\footnotesize BA-SEQ-IDEM}\\
p\cdot 1 & \,\equiv\, & p & \textnormal{\footnotesize KA-SEQ-ONE}~ & & &  & \\
p\cdot (q+r) & \,\equiv\, & p\cdot q + p\cdot r & \textnormal{\footnotesize KA-SEQ-DIST-L}~ & f \leftarrow n \cdot f' \leftarrow n' & \,\equiv\, & f' \leftarrow n' \cdot f \leftarrow n, \textnormal{if } f \not = f' ~~& \textnormal{\footnotesize PA-MOD-MOD-COMM}\\
(p+q)\cdot r & \,\equiv\, & p\cdot r + q\cdot r & \textnormal{\footnotesize KA-SEQ-DIST-R}~ & f \leftarrow n \cdot f' = n' & \,\equiv\, & f' = n' \cdot f \leftarrow n, \textnormal{if } f \not = f' & \textnormal{\footnotesize PA-MOD-FILTER-COMM} \\
0\cdot p & \,\equiv\, & 0 & \textnormal{\footnotesize KA-ZERO-SEQ}~ & {\bf dup} \cdot f = n & \,\equiv\, & f = n \cdot {\bf dup} & \textnormal{\footnotesize PA-DUP-FILTER-COMM} \\
p\cdot 0 & \,\equiv\, & 0 & \textnormal{\footnotesize KA-ZERO-SEQ}~ & f \leftarrow n \cdot f = n & \,\equiv\, & f \leftarrow n & \textnormal{\footnotesize PA-MOD-FILTER} \\
1 + p\cdot p^* & \,\equiv\, & p^* & \textnormal{\footnotesize KA-UNROLL-L}~ & f = n \cdot f \leftarrow n & \,\equiv\, & f = n & \textnormal{\footnotesize PA-FILTER-MOD} \\
1 + p^*\cdot p & \,\equiv\, & p^* & \textnormal{\footnotesize KA-UNROLL-R}~ & f \leftarrow n \cdot f \leftarrow n' & \,\equiv\, & f \leftarrow n' & \textnormal{\footnotesize PA-MOD-MOD} \\
q+p\cdot r \leq r  & \,\Rightarrow\, & p^* \cdot q\leq r & \textnormal{\footnotesize KA-LFP-L}~ & f = n \cdot f = n' & \,\equiv\, &\,0, \textnormal{if } n \not= n'& \textnormal{\footnotesize PA-CONTRA} \\
p+q\cdot r \leq q  & \,\Rightarrow\, & p\cdot r^*\leq q & \textnormal{\footnotesize KA-LFP-R}~ & \Sigma_i f = i & \,\equiv\, & \,1 & \textnormal{\footnotesize PA-MATCH-ALL} 
\end{array}}
}
\]
\caption{NetKAT Axiomatisation~\cite{DBLP:conf/popl/AndersonFGJKSW14}} \label{fig::axioms-NetKAT} 
\end{figure}
\color{black}

NetKAT \emph{packets} ${\it pk}$ are encoded as sets of fields $f_i$ and associated values $v_i$ as in Figure~\ref{fig::netKAT-syn-sem}. \emph{Histories} are defined as lists of packets, and are exploited in order to define the semantics of NetKAT policies/programs as in Figure~\ref{fig::netKAT-syn-sem}. NetKAT \emph{policies} are recursively defined as: {predicates}, field modifications $f \leftarrow n$, union of policies $p+q$ ($+$ plays the role of a multi-casting like operator), sequencing of policies $p\cdot q$, repeated application of policies $p^*$ (the Kleene $*$) and duplication $\bf dup$ (that saves the current packet at the beginning of the history list).
At this point, it might be worth mentioning that {\bf dup} plays a role in building the NetKAT language model but, as we shall later see, it is not necessary in our syntactic approach to failure analysis.

\emph{Predicates}, on the other hand, can be seen as filters. The constant predicate $0$ drops all the packets, whereas its counterpart predicate $1$ retains all the packets. The test predicate $f = n$ drops all the packets whose field $f$ is not assigned value $n$. Moreover, $\neg a$ stands for the negation of predicate $a$, $a+b$ represents the disjunction of predicates $a$ and $b$, whereas $a\cdot b$ denotes their conjunction.

Let $H$ be the set of all histories, and ${\cal P}(H)$ be the power set of $H$.
In Figure~\ref{fig::netKAT-syn-sem},
the semantic definition of a NetKAT policy $p$ is given as a function $\llbracket p\rrbracket$ that takes a history $h \in H$ and produces a (possibly empty) set of histories in ${\cal P}(H)$.
Some intuition on the semantics of policies was already provided in the paragraph above.
In addition, note that negated predicates drop the packets not satisfying that predicate: $\llbracket \neg a\rrbracket h = \{h\} \setminus \llbracket a \rrbracket h$. The sequential composition of policies $\llbracket p \cdot q\rrbracket$ denotes the Kleisli composition $\bullet$ of the functions  $\llbracket p \rrbracket$ and  $\llbracket q\rrbracket$.

The repeated iteration of policies is interpreted as the union of $F^i\,h$, where the semantics of each $F^i $ coincides with the semantics of the policy resulted by concatenating $p$ with itself for $i$ times, for $i \in {\mathbb N}$.

In Figure~\ref{fig::axioms-NetKAT} we recall the sound and complete axiomatisation of NetKAT. The Kleene Algebra with Tests axioms in Figure~\ref{fig::axioms-NetKAT}, have been formerly introduced in~\cite{DBLP:journals/toplas/Kozen97}.
Completeness of NetKAT results from the packet algebra axioms in Figure~\ref{fig::axioms-NetKAT}.
The axiom {\small PA-MOD-MOD-COMM} stands for the commutativity of different field assignments, whereas {\small PA-MOD-FILTER-COMM} denotes the commutativity of different field assignments and tests, for instance.
{\small PA-MOD-MOD} states that two subsequent modifications of the same field can be reduced to capture the last modification only. The axiom {\small PA-CONTRA} states that the same field of a packet cannot have two different values, etc.

We write $\vdash e \equiv e'$, or simply $e \equiv e'$, whenever the equation $e \equiv e'$ can be proven according to the NetKAT axiomatisation.

Assume, for an example, a simple network consisting four hosts $H_1, H_2, H_3$ and $H_4$ communicating with each other via two switches $A$ and $B$, via the uniquely-labeled ports $1, 2, \ldots, 6$, as illustrated in Figure~\ref{fig::simpl-net}. The network topology can be given by the NetKAT expression:
\begin{equation}\label{eq::topology}
\begin{array}{rcl}
t & \triangleq & \pt= 5 \cdot \pt \leftarrow 6 + \pt = 6 \cdot \pt \leftarrow 5 \,+\\
& & \pt=1 + \pt=2 + \pt=3 + \pt=4
\end{array}
\end{equation}
For an intuition, in~(\ref{eq::topology}), the expression $ \pt= 5 \cdot \pt \leftarrow 6 + \pt = 6 \cdot \pt \leftarrow 5$ encodes the internal link $5 - 6$ by using the sequential composition of a filter that keeps the packets at one end of the link and a modification that updates the $\pt$ fields to the location at the other end of the link.
A link at the perimeter of the network is encoded as a filter that returns the packets located at the ingress port.


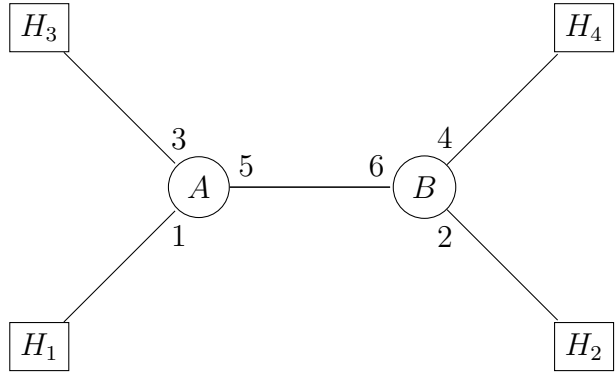
\begin{figure}[h] 
               \centering \begin{tikzpicture}[>=stealth,
   shorten >=1pt,auto,node distance=3cm,switch node/.style={circle,draw}, host node/.style={rectangle,draw}]
                  \node[host node] (h1) {$H_1$};

                  \node[switch node] (s1) [above right of=h1] {$A$};
                  \node[host node] (h3) [above left of=s1] {$H_3$};
                  \node[switch node] (s2) [right of=s1] {$B$};
                  \node[host node] (h2) [below right of=s2] {$H_2$};
                  \node[host node] (h4) [above right of=s2] {$H_4$};
                  
                  \path[-]
                    (h1) edge [pos=0.95] node [xshift=0.1cm,below] {$1$} (s1)
                    (h3) edge [pos=0.95] node [xshift=0.1cm,above] {$3$} (s1)
                    (s1) edge [pos=0.1] node [above] {$5$} (s2)
                    (s1) edge [pos=0.9] node [above] {$6$} (s2)
                    (s2) edge [pos=0.05] node [xshift=-0.1cm,above] {$4$} (h4)
                    (s2) edge [pos=0.05] node [xshift=-0.1cm,below] {$2$} (h2);

                \end{tikzpicture}
            \caption{A Simple Network} \label{fig::simpl-net}
            \end{figure}

Furthermore, assume a programmer $P_1$ as in~\cite{DBLP:conf/popl/AndersonFGJKSW14} which has to encode a switch policy that only enables transferring packets from $H_1$ to $H_2$. $P_1$ might define the ``hop-by-hop'' policy in~(\ref{eq::pol-p1}), where each summand stands for the forwarding policy on switch $A$ and $B$, respectively. 
\begin{equation}\label{eq::pol-p1}
p_1  \triangleq \pt = 1 \cdot \pt \leftarrow 5 + \pt = 6 \cdot \pt \leftarrow 2
\end{equation}
In the expression above, the NetKAT expression $\pt = 1 \cdot \pt \leftarrow 5$ sends the packets arriving at port $1$ on switch $A$, to port $5$, whereas $\pt = 6 \cdot \pt \leftarrow 2$ sends the packets at port $6$ on switch $B$, to port $2$. 

At this point, from $P_1$'s perspective, the end-to-end behaviour of the network is defined as:
\begin{equation}\label{eq::p1-e-to-e}
(\pt = 1) \cdot (p_1 \cdot t)^* \cdot (\pt = 2)
\end{equation}
In words: packets situated at ingress port $1$ (encoded as $\pt = 1$) are forwarded to egress port $2$ (encoded as $\pt = 2$) according to the switch policy $p_1$ and topology $t$ (encoded as $(p_1 \cdot t)^*$).

More generally, assuming a switch policy $p$, topology $t$, ingress ${\it in}$ and egress ${\it out}$, the \emph{end-to-end behaviour} of a network is defined as:
\begin{equation}\label{eq::net-beh}
{\it in} \cdot (p \cdot t)^* \cdot {\it out}
\end{equation}

Note that, unlike the end-to-end NetKAT network behaviour in~\cite{DBLP:conf/popl/AndersonFGJKSW14}, the policy in~(\ref{eq::net-beh}) does not contain {\bf dup}. As discussed in more detail in Section~\ref{sec::expln-fail}, our (syntactic) approach looks at each operation within a NetKAT expression, hence there is no need to use {\bf dup} in order to record the individual ``hops'' that packets take as they go through the network.

Based on~(\ref{eq::p1-e-to-e}), in order to assess the correctness of $P_1$'s program, one has to show that:
\begin{enumerate}
\item packets at port $1$ reach port $2$, i.e.,
\begin{equation}\label{eq::p1-reach}
\vdash (\pt = 1) \cdot (p_1 \cdot t)^* \cdot (\pt = 2) \not \equiv 0
\end{equation}
\item no packets at port $1$ can reach ports $3$ or $4$, i.e.,
\begin{equation}\label{eq::p1-not-reach}
\vdash (\pt = 1) \cdot (p_1 \cdot t)^* \cdot (\pt = 3 + \pt = 4) \equiv 0.
\end{equation}
\end{enumerate}

By applying the NetKAT axiomatisation, the inequality in~(\ref{eq::p1-reach}) can be equivalently rewritten as:
\begin{equation}\label{eq::equiv-one-reach-two}
\vdash \pt = 1 \cdot \pt \leftarrow 2 + e\not \equiv 0
\end{equation}
with $e$ a NetKAT expression. Observe that $\pt = 1 \cdot \pt \leftarrow 2 $ cannot be reduced further. Hence, the inequality in~(\ref{eq::p1-reach}) holds, as $\pt = 1 \cdot \pt \leftarrow 2 \not \equiv 0$. In other words, the packets located at port $1$ reach port $2$.
Showing that no packets at port $1$ can reach port $3$ or $4$ follows in a similar fashion.

\section{Safety and Failures in NetKAT}
\label{sec::safety}

As discussed in the previous section, arguing on equivalence of NetKAT programs can be easily performed in an equational fashion. One interesting way of further exploiting the NetKAT framework is to formalise and reason about well-known notions of program correctness such as safety, for instance. Intuitively, a safety property states that ``something bad never happens''. Ideally, the framework would provide a positive answer whenever a certain safety property is satisfied by the program, and an explanation of what went wrong in case the property is violated.

Consider the example of programmer $P_1$. The ``bad'' thing that could happen is that his switch policy enabled packets to reach ports $3$ or $4$. One can encode such a hazard via the egress policy ${\it out} \triangleq \pt = 3 + \pt = 4$, and the whole safety requirement as in~(\ref{eq::p1-not-reach}). As previously discussed, the NetKAT axiomatisation provides a positive answer with respect to the satisfiability of the safety requirement in~(\ref{eq::p1-not-reach}).

Firstly, observe that our approach is syntactic in nature and it does not require recording individual packet modifications, or simulating actual ``moves" in the NetKAT corresponding automata. Hence, it suffices to consider $\dup$-free NetKAT expressions. As we shall later see, this also contributes to deriving more concise, $\dup$-free failure explanations.

Secondly, observe that from a more practical perspective, the Kleene-$*$ is mainly used for ensuring a ``looping" structure to allow packet moves along the hops. Thus, in our work, we consider ingress ($\inn$), egress ($\out$), switch policies ($p$) and topologies ($t$) encoded in terms of $\dup$-free, $*$-free NetKAT expressions, while the overall behaviour of a network is given as $\inn \cdot (p \cdot t)^* \cdot \out$.

We call $\NetKATdup$ the $\dup$-free, $*$-free fragment of NetKAT. We further proceed by formalizing a safety concept in NetKAT.

\begin{definition}[In-Out Safe]\label{def::safety}
Assume the $\NetKATdup$ expressions defining a network topology $t$, a switch policy $p$, an ingress policy $\inn$, and an egress policy $\out$, the latter encoding the hazard, or the ``bad thing''. The end-to-end network behaviour is \emph{in-out safe} whenever the following holds:
\begin{equation}
\vdash \inn \cdot (p \cdot t)^* \cdot \out \equiv 0.
\end{equation}
\end{definition}
Intuitively, none of the packages at ingress $\inn$ can reach the ``hazardous'' egress $\out$ whenever forwarded according to the switch policy $p$, across the topology $t$.

We call the \emph{size of the network} the number of forwarding links within the network.

\begin{remark}
A notion of reachability within NetKAT-definable networks was proposed in~\cite{DBLP:conf/popl/AndersonFGJKSW14} based on the existence of a non-empty packet history that, in essence, records all the packet modifications produced by the policy $ \inn \cdot (p \cdot t)^* \cdot \out $. This is more like a model-checking-based technique that enables identifying \emph{one} counterexample witnessing the violation of the property $ \inn \cdot (p \cdot t)^* \cdot \out \equiv 0$.
As we shall later see, in our setting, we are interested in identifying \emph{all} (minimal) counterexamples. Hence, we propose a notion of in-out safe behaviour for which, whenever violated, we can provide all relevant bad behaviours.
\end{remark}

Going back to the example in Section~\ref{sec::prelim}, assume a new programmer $P_2$ which has to enable traffic only from $H_3$ to $H_4$. Assuming the network in Figure~\ref{fig::simpl-net}, $P_2$ encodes the {\PHbH} switch policy:
\begin{equation}\label{eq::pol-p2}
p_2  \triangleq \pt = 3 \cdot \pt \leftarrow 5 + \pt = 6 \cdot \pt \leftarrow 4
\end{equation}
The end-to-end behaviour can be proven correct, by showing that:
\begin{enumerate}
\item packets at port $3$ reach port $4$, i.e.,
\begin{equation}\label{eq::p2-reach}
\vdash (\pt = 3) \cdot (p_2 \cdot t)^* \cdot (\pt = 4) \not \equiv 0
\end{equation}
\item no packets at port $3$ can reach ports $1$ or $2$, i.e.,
\begin{equation}\label{eq::p2-not-reach}
\vdash (\pt = 3) \cdot (p_2 \cdot t)^* \cdot (\pt = 1 + \pt = 2) \equiv 0.
\end{equation}
\end{enumerate}

Nevertheless, it is easy to show that the composed policies $p_1$ in~(\ref{eq::pol-p1}) and $p_2$ in~(\ref{eq::pol-p2}) do not guarantee a safe behaviour. Namely, in the context of the {\PHbH} policy $p_1 + p_2$, packets at port $1$ can reach port $4$, and packets at port $3$ can reach port $2$. This violates the correctness properties in~(\ref{eq::p1-not-reach}) and~(\ref{eq::p2-not-reach}), respectively:
\begin{equation}\label{eq::p12-not-safe1}
\vdash (\pt = 1) \cdot ((p_1 + p_2) \cdot t)^* \cdot (\pt = 3 + \pt = 4)\not \equiv 0
\end{equation}
\begin{equation}\label{eq::p12-not-safe2}
\vdash (\pt = 3) \cdot ((p_1 + p_2) \cdot t)^* \cdot (\pt = 1 + \pt = 2)\not \equiv 0
\end{equation}

In the next section, we provide a framework for explaining the failure of network safety as expressed in~(\ref{eq::p12-not-safe1}) and~(\ref{eq::p12-not-safe2}).

\subsection{Explaining Safety Failures}\label{sec::expln-fail}

Naturally, the first attempt to explain safety failures is to derive the counterexamples according to the NetKAT axiomatisation. Take, for instance, the end-to-end behaviour $(\pt = 1) \cdot ((p_1 + p_2) \cdot t)^* \cdot (\pt = 3 + \pt = 4)$ in~(\ref{eq::p12-not-safe1}). The axiomatisation leads to the following equivalence:
\begin{equation}\label{eq::simple-counterexample}
(\pt = 1) \cdot ((p_1 + p_2) \cdot t)^* \cdot (\pt = 3 + \pt = 4) \equiv (\pt = 1 \cdot \pt \leftarrow 4) + e
\end{equation}
where $e$ is a NetKAT expression containing the Kleene $*$. A counterexample can be immediately spotted, namely: $\pt = 1 \cdot \pt \leftarrow 4$. Nevertheless, the information it provides is not intuitive enough to serve as an explanation of the failure. Moreover, $e$ can hide additional counterexamples revealed after a certain number of $*$-unfoldings according to {\small KA-UNROLL-R} and {\small KA-UNROLL-L} in Figure~\ref{fig::axioms-NetKAT}.

In what follows, the focus is on the following two questions:
\begin{enumerate}
\item[$Q_1$:] Can we reveal \emph{more information} within the counterexamples witnessing safety failures?
\item[$Q_2$:] Can we reveal \emph{all the counterexamples} hidden within NetKAT expressions containing $*$?
\end{enumerate}

The answer to $Q_1$ is relatively simple: yes, we can reveal more information on how the packets travel across the topology by removing the {\small PA-MOD-MOD} and {\small PA-FILTER-MOD} axioms in Figure~\ref{fig::axioms-NetKAT}. Recall that, intuitively, this axiom records only the last modification from a series of modifications of the same field.

The answer to $Q_2$ lies behind the following two observations. (1) From a practical perspective, in order to explain failures it suffices to look at minimal forwarding paths within the network topology that lead from $\inn$ to $\out$. (2) Traversing the same path twice does not add insightful information about the reason behind the violation of a safety property, as the network behaviour is preserved in the context of that path. This is also in accordance with the minimality criterion invoked in the seminal work on causal reasoning in~\cite{DBLP:conf/sum/Halpern11}, for instance.
It is intuitive to see that given a NetKAT program $\inn \cdot (p \cdot t)^*  \cdot \out$ there is a sufficient number of $*$-unfoldings that can reveal all the relevant paths from $\inn$ to $\out$. As shown by our experimental evaluation, in most of the practical cases, it suffices to
analyze paths of length
equal with the size $n$ of the network.

Theorem~\ref{thm::star-elim-safety} states that safety in NetKAT programs reduces to showing that there are no paths from $\inn$ to $\out$ for any hop-by-hop forwarding strategy on individual switches complying to a switch policy $p$. The result in Theorem~\ref{thm::star-elim-safety} follows straightforwardly by Lemma~\ref{lm::expand-m} and Lemma~\ref{lm::expand-m-star}.

Given a NetKAT policy $q$ and a natural number $m$, we write $q^m$ to denote the repeated application of $q$ for $m$ times:
\[
q^m = \begin{cases}
    1, & \text{if $m = 0$}\\
    q \cdot q ^ {m-1} , & \text{if $m \geq 1$}.
  \end{cases}
\]
We call \emph{repetitions} expressions of shape $p^m$.





\begin{lemma}\label{lm::expand-m}
Let $p,\, t$ be two NetKAT policies. The following holds:
\begin{equation}\label{eq::unfold}
\forall n \in \mathbb{N}.~ (1 + p\cdot t)^n \,\, \equiv \,\, 1 + p\cdot t + (p\cdot t)^2 + \ldots + (p\cdot t)^n
\end{equation}
\end{lemma}
\begin{proof}
The proof follows immediately, by induction on $n$ and by the Kleene Algebra axioms in Figure~\ref{fig::axioms-NetKAT}.\\
\emph{Base case:} $n = 0$. If $n = 0$ then $(1 + (p \cdot t))^0 = 1$, inferred based on the definition of Kleisli composition.\\
\emph{Induction step:} Assume~(\ref{eq::unfold}) holds for all $k$ such that $0 \leq k \leq n$. It follows that:

\[
\begin{array}{rl}
(1 + p\cdot t)^{n+1} & \equiv_\textnormal{(Kleisli comp.)}\\
(1 + p\cdot t)^n \cdot (1 + p\cdot t) & \equiv_\textnormal{(ind. hypo.)}\\
(1 + p\cdot t + (p\cdot t)^2 + \ldots + (p\cdot t)^n)\cdot (1 + p\cdot t) & \equiv_\textnormal{({ KA-SEQ-DIST-L/R, }}\\
 & {~}_\textnormal{~~~KA-PLUS-IDEM)}\\
1 + p\cdot t + (p\cdot t)^2 + \ldots + (p\cdot t)^n + &\\
p\cdot t + (p\cdot t)^2 + \ldots + (p\cdot t)^n + (p\cdot t)^{n+1} & \equiv_\textnormal{({KA-PLUS-IDEM})}\\
1 + p\cdot t + (p\cdot t)^2 + \ldots + (p\cdot t)^{n} + (p\cdot t)^{n+1} &
\end{array}
\]
Hence,~(\ref{eq::unfold}) holds.
\end{proof}

\begin{lemma}\label{lm::expand-m-star}
Let $p,\, t,\, \inn,\, \out$ be NetKAT policies. The following holds:
\begin{equation}\label{eq::unfold-star}
\forall n \in \mathbb{N}.~\inn \cdot (1 + p\cdot t)^n\cdot \out \,\, \leq \,\, \inn \cdot (p\cdot t)^*\cdot \out
\end{equation}
\end{lemma}
\begin{proof}
%
Consider $n \in \mathbb{N}$. First, observe that
\begin{equation}\label{eq:rew-star-n-plus-one}
\begin{array}{l}
\inn \cdot (p\cdot t)^*\cdot \out \equiv\\
\inn \cdot (1 + p\cdot t + (p\cdot t)^2 + \ldots + (p\cdot t)^{n} + (p\cdot t)^{n+1}\cdot (p\cdot t)^* )\cdot \out 
\end{array}
\end{equation}
by {\small KA-UNROLL-L, KA-UNROLL-R, KA-PLUS-IDEM} and {\small KA-SEQ-DIST-L, KA-SEQ-DIST-R}.
Consequently, by Lemma~\ref{lm::expand-m}, the following also holds:
\begin{equation}\label{eq:rew-star-n-plus-one-lm1}
\inn \cdot (p\cdot t)^* \cdot \out \equiv \inn \cdot (1 + p\cdot t)^{n} \cdot \out +  \inn \cdot (p\cdot t)^{n+1}\cdot (p\cdot t)^*\cdot \out
\end{equation}
Therefore,
\[
\inn \cdot (1 + p\cdot t)^{n}\cdot \out \,\, \leq \,\, \inn \cdot(p\cdot t)^*.\out
\]
holds by the definition of the partial order relation $\leq$.
\end{proof}

\begin{theorem}(Approximation Principle for Safety)\label{thm::star-elim-safety}
Assume a network topology $t$, a switch policy $p$, an ingress policy $\inn$, and an egress policy $\out$ encoding the hazard. The following holds:
\begin{equation}\label{eq:star-elim-safe}
\vdash \inn \cdot (p\cdot t)^*\cdot \out \equiv 0 {\textnormal{~~iff~~}} \forall n \in \mathbb{N} . \vdash \inn \cdot (1 + p\cdot t)^n\cdot \out \equiv 0
\end{equation}
\end{theorem}
\begin{proof}
The ``if'' case follows immediately, as by Lemma~\ref{lm::expand-m-star}, the hypothesis $\inn \cdot (p\cdot t)^*\cdot \out \equiv 0$ and the fact that $0 \leq q$ for all NetKAT policies $q$, the following holds:
\[
\forall n \in \mathbb{N}.~0 \leq \inn \cdot (1 + p\cdot t)^{n}\cdot \out \,\, \leq \,\, \inn \cdot (p\cdot t)^*\cdot \out \equiv 0.
\]

For the ``only if'' case we proceed by reductio ad absurdum.

Assume $\forall n \in \mathbb{N} . \vdash \inn \cdot (1 + p\cdot t)^n\cdot \out \equiv 0$ and
\begin{equation}\label{eq::red-abs}
\inn \cdot (p\cdot t)^*\cdot \out \not \equiv 0.
\end{equation}
By the definition of the Kleene $*$ and the assumption in~(\ref{eq::red-abs}), it follows that there exists $m \in \mathbb{N}$ such that:
\[
\inn \cdot (p\cdot t)^m\cdot \out \not \equiv 0.
\]
By Lemma~\ref{lm::expand-m}, we can see that the latter contradicts the hypothesis. Hence, our assumption is false.
\end{proof}

\begin{remark}[Construction of $\svdash$]\label{rm::new-entail}
With these ingredients at hand, in accordance with $Q_1$ and $Q_2$, we consider an alteration of the NetKAT axiomatisation.
Recall that our NetKAT policies do not use {\bf dup}. Our approach is purely syntactic (it does not involve network packet analysis) and 
it looks at each operation within a NetKAT expression, in a ``small-step" fashion. This can be achieved by removing the axioms {\small PA-MOD-MOD} and {\small PA-FILTER-MOD}.

Let $\svdash$ be the new entailment relation over the modified axiomatisation.
\end{remark}

\begin{remark}\label{rm::compl-lost}
Note that $\svdash$ is no longer complete. Nevertheless, the purpose of $\svdash$ is not to prove equivalence of arbitrary $\NetKATdup$, but to identify safety failure violations and corresponding explanations. In what follows, we show a series of useful/interesting properties of $\svdash$.
\end{remark}

\begin{theorem}[Consistency of $\svdash$]\label{def::consistent-entail}
Assume a $\NetKATdup$ policy p. The following holds:
\begin{equation}\label{eq:consistent-entail}
\vdash p \equiv 0 \textnormal{ iff } \svdash p \equiv 0
\end{equation}
\end{theorem}
\begin{proof}
The key observation behind this proof is that $0$-terms can only be derived according to the {\small BA/PA-CONTRA} axioms:
\[
\begin{array}{rcll}
a \cdot \neg a & \equiv & 0 &\\
f = n \cdot f = n' & \equiv & 0 & \textnormal{ if } n \not = n'
\end{array}
\]
The removed axiom {\small PA-MOD-MOD}
\[
f \leftarrow n \cdot f \leftarrow n' \equiv f \leftarrow n'
\]
can only involve tests when used in combination with the {\small PA-MOD-FILTER} axiom:
\[
f \leftarrow n \cdot f = n \equiv f \leftarrow n
\]
This implies:
\[
f \leftarrow n \cdot f \leftarrow n' \equiv f \leftarrow n \cdot f = n \cdot f \leftarrow n' \cdot f = n'
\]
Nevertheless, the right hand side of the above reduction can never be evaluated to $0$ as commutativity of $\leftarrow$ and $=$ is only allowed in the context of different fields, according to {small PA-MOD-FILTER-COMM}:
\[
f \leftarrow n \cdot f' = n'  \equiv  f' = n' \cdot f \leftarrow n \textnormal{~~if } f \not = f'
\]

Moreover, it is straightforward to see that {\small PA-FILTER-MOD}
\[
f = n \cdot f \leftarrow n \equiv f = n
\]
has no influence on the evaluation to $0$-terms, as tests are not removed by this axiom.

It is, therefore, safe to conclude that~(\ref{eq:consistent-entail}) holds.
\end{proof}

Hence, according to Theorem~\ref{thm::star-elim-safety} and Theorem~\ref{def::consistent-entail}, we can conclude that a network behaviour is ``in-out-safe'' whenever it can be proven so according to $\svdash$:

\begin{corollary}[Safety Sound \& Complete]\label{cor::safe-compl}
Assume the $\NetKATdup$ policies encoding a network topology $t$, a switch policy $p$, an ingress policy $\inn$, and an egress policy $\out$ encoding the hazard. The following holds:
\begin{equation}\label{eq:consistent-entail2}
\vdash \inn \cdot (p\cdot t)^*\cdot \out \equiv 0 \textnormal{ iff }
\forall n \in \mathbb{N}.~\svdash \inn \cdot (1 + p\cdot t)^n\cdot \out \equiv 0
\end{equation}
\end{corollary}

As previously stated, our experimental evaluation showed that in most of the cases it suffices to consider a limited number of $*$-unfoldings equal to the size $n$ of the network, in order to reveal all the possible ways of reaching a hazardous egress $\out$ from a given ingress $\inn$. In accordance, we introduce a notion of so-called $n$-safety failure explanations.

\begin{definition}[$n$-Safety Failure Explanations]\label{def::safety-fail-expl}
Assume the $\NetKATdup$ policies encoding
a network topology $t$, a switch policy $p$, an ingress policy $\inn $, and an egress policy $\out$ encoding the hazard. An \emph{$n$-safety failure explanation} is a policy {${\it expl} \not \equiv 0$} such that, for $n\in \mathbb{N}$:
\begin{equation}\label{eq:star-elim-safe2}
\svdash \inn \cdot (1 + p\cdot t)^n\cdot \out \equiv {\it expl}.
\end{equation}
\end{definition}

For an example, we refer to the case of the two programmers providing switch policies $p_1$ and $p_2$ forwarding packets from host $H_1$ to $H_2$, and from $H_3$ to $H_4$ within the network in Figure~\ref{fig::simpl-net}. As previously discussed, the end-to-end network behaviour defined over each of the aforementioned policies can be proven correct using the NetKAT axiomatisation. Nevertheless, a comprehensive explanation of what caused the erroneous behaviour over the unified policy $p_1 + p_2$ could not be derived according $\vdash$. Note that the network consists of $6$ forwarding links. Hence, $6$ unfoldings were sufficient for the new axiomatisation to entail the following explanation:
\[
\svdash (\pt = 1) \cdot ((p_1 + p_2) \cdot t)^6 \cdot (\pt = 3 + \pt = 4) \equiv \pt = 1 \cdot \pt \leftarrow 5 \cdot \pt \leftarrow 6 \cdot \pt \leftarrow 4
\]
showing how packets at port $1$ can reach port $4$.
Similarly,
\[
\svdash (\pt = 3) \cdot ((p_1 + p_2) \cdot t)^6 \cdot (\pt = 1 + \pt = 2) \equiv \pt = 3 \cdot \pt \leftarrow 5 \cdot \pt \leftarrow 6 \cdot \pt \leftarrow 2
\]
shows how packets at port $3$ can reach port $2$. 

\begin{remark}\label{rm::no-dups}
The work in~\cite{DBLP:conf/popl/AndersonFGJKSW14} proposes a ``star elimination'' method for switch policies not containing $\bf{dup}$ and switch assignments. The procedure in~\cite{DBLP:conf/popl/AndersonFGJKSW14} employs a notion of normal form to which each NetKAT policy can be reduced. The reason for not using the aforementioned star elimination in our context is that the normal forms in~\cite{DBLP:conf/popl/AndersonFGJKSW14} ``forget'' the intermediate sequences of assignments and tests, and reduce policies to sums of expressions of shape $(f_1 = v_1 .\, \ldots \,. f_n = v_n) \cdot (f_1 \leftarrow v'_1 .\, \ldots \,. f_n \leftarrow v'_n)$ where $f_1, \ldots, f_n$ are the packet fields. Hence, the normal forms exploited by the star elimination in~\cite{DBLP:conf/popl/AndersonFGJKSW14} can not serve as comprehensive failure explanations.
\end{remark}

\begin{figure}[h] 
               \centering \begin{tikzpicture}[>=stealth,
   shorten >=1pt,auto,node distance=4cm,switch node/.style={circle,draw}, host node/.style={rectangle,draw}]
                  \node[host node] (h1) {$H_1$};
                  \node[switch node] (s1) [right of=h1] {$A$};
                  \node[switch node] (s2) [above of=s1] {$B$};
                  \node[host node] (h2) [right of=s1] {$H_2$};
                  
                  \path[->]
                    (h1) edge [pos=0.9] node [yshift=-0.1cm, below] {$1$} (s1)
                    (s1) edge [pos=0.05] node [yshift=-0.1cm, below] {$3$} (h2);

                \path[->, transform canvas={shift={(-0.1,0)}}]
                    (s1) edge [pos=0.1] node [left] {$2$} (s2)
                    (s1) edge [pos=0.9] node [left] {$4$} (s2);

                \path[->, transform canvas={shift={(0.1,0)}}]
                    (s2) edge [pos=0.1] node [right] {$5$} (s1)
                    (s2) edge [pos=0.9] node [right] {$6$} (s1);

                \end{tikzpicture}
            \caption{A Firewall} \label{fig::firewall-example}
            \end{figure}
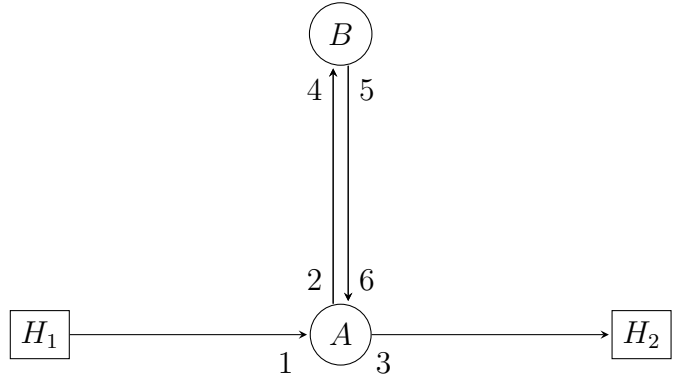
            
We next provide an additional firewall example to better illustrate the ideas in Remark~\ref{rm::no-dups}. 
Consider a scenario where there are two hosts $H_1$ and $H_2$, a switch $A$, and a firewall $B$, as displayed in Figure~\ref{fig::firewall-example}. In this setting the packets that reach $A$ are first forwarded to the firewall, and then to their destination, and the firewall blocks all non-SSH traffic. The policy and the topology are defined as follows:

\[
\begin{array}{ll}
 p \triangleq & sw = A \cdot (dst = H_2 \cdot firewalled = 0 \cdot pt \leftarrow 2 + \\
 &  ~~~~~~~~~~~~~dst = H_2 \cdot firewalled = 1 \cdot pt \leftarrow 3) + \\
 &  sw = B \cdot (typ = SSH \cdot firewalled \leftarrow 1 \cdot pt \leftarrow 5) \\ \\
 t \triangleq & sw = A \cdot (pt = 2 \cdot sw \leftarrow B \cdot pt \leftarrow 4 + pt = 1 + pt = 3) + \\
& sw = B \cdot pt = 5 \cdot sw \leftarrow A \cdot pt \leftarrow 6
\end{array}
\]

Assume that packets from $H_1$ reaching to $H_2$ constitutes a safety violation. The $in$ and $out$ are defined as follows:
\begin{alignat}{2}
in \triangleq &~sw = A \cdot pt = 1 \cdot dst = H_2 \cdot firewalled = 0 \notag \\
out \triangleq &~sw = A \cdot pt = 3\notag
\end{alignat}

Generally speaking, we are interested to check whether $\inn \cdot (p \cdot t)^* \cdot \out$ reduces to $0$ (indicating the absence of the hazard) or not. Based on the framework devised in this paper, this reduces to checking the aforementioned equalities after unfolding the expression $(p \cdot t)^*$ for a number of times equal to the number of (oriented) links in the network.
It is clear that in our case we are interested to check whether $in \cdot (p \cdot t)^4 \cdot out \equiv 0$ or not. Our framework gives the following counterexample:   \begin{alignat}{2}
& sw = A \cdot pt = 1 \cdot dst = H_2 \cdot  firewalled = 0 \cdot\notag \\ 
& pt \leftarrow 2 \cdot sw \leftarrow B \cdot 
pt \leftarrow 4 \cdot \notag \\ 
&  typ = SSH \cdot firewalled \leftarrow 1 \cdot pt \leftarrow 5 \cdot sw \leftarrow A \cdot  pt \leftarrow 6 \cdot \notag \\
& pt \leftarrow 3 \notag
\end{alignat}

\begin{remark}
In~\cite{DBLP:conf/popl/AndersonFGJKSW14}, the completeness theorem of NetKAT is based on a language model:
\begin{equation}
    \alpha \cdot \pi_0 \cdot~\mathbf{dup}~\cdot \pi_1 \cdot~\mathbf{dup}~ \ldots~\mathbf{dup} \cdot \pi_n
\end{equation}
where $\alpha \triangleq f_1=n_1 \ldots f_k=n_k$ is called a complete test and $\pi \triangleq f_1\leftarrow n_1 \ldots f_k\leftarrow n_k$ is called a complete assignment. Note that the axiom that we removed, \small{PA-MOD-MOD}, plays an important role in bringing the expressions into this form. If we had strictly followed the approach in~\cite{DBLP:conf/popl/AndersonFGJKSW14}, then for the above firewall example we would have obtained a counterexample of the following shape:
\begin{equation}\label{eq::long-conter}
\begin{array}{l}
(sw = A \cdot pt = 1 \cdot dst = H_2 \cdot typ = SSH \cdot firewalled = 0) \cdot \\
(sw \leftarrow A \cdot pt \leftarrow 1 \cdot dst \leftarrow H_2 \cdot typ \leftarrow SSH \cdot firewalled \leftarrow 0) \cdot \mathbf{dup} \cdot  \\
(sw \leftarrow B \cdot pt \leftarrow 4 \cdot dst \leftarrow H_2  \cdot typ \leftarrow SSH \cdot firewalled \leftarrow 0) \cdot \mathbf{dup} \cdot  \\
(sw \leftarrow A \cdot pt \leftarrow 6 \cdot dst \leftarrow H_2 \cdot typ \leftarrow SSH \cdot firewalled \leftarrow 1) \cdot \mathbf{dup} \cdot \\
(sw \leftarrow A \cdot pt \leftarrow 3 \cdot dst \leftarrow H_2 \cdot typ \leftarrow SSH \cdot firewalled \leftarrow 1)
\end{array}
\end{equation}
Observe that a more concise, $\dup$-free counterexample is obtained from our approach, which we believe is better suitable in the context of causality checking. Furthermore, certain information has been lost in the expression in~(\ref{eq::long-conter}), $i.e.$ the assignments $pt \leftarrow 2$ and $pt \leftarrow 5$ do not appear in the counterexample. More generally, if there exist more than one assignment to a field inside $p \cdot t$, then only the last assignment is preserved. We believe this is not favorable for causality checking.
\end{remark}

\subsection{Minimal Explanations}\label{sec:minimality}

Note that the safety failure explanations in Definition~\ref{def::safety-fail-expl} are not minimal.
For an example, there might be cases in which two explanation paths of shape
\[
\begin{array}{rl}
e_1 \triangleq p'\cdot p'' ~~~~~~~~~~~~ e_2 \triangleq p'\cdot \tilde{p}\cdot p''
\end{array}
\]
are identified. In this case, we consider $e_1$ as more ``expressive" than $e_2$. In this section we introduce a notion of minimality, inspired by the seminal works on causal reasoning in~\cite{DBLP:conf/sum/Halpern11,DBLP:conf/ijcai/Halpern15}.
We define minimality based on a notion of NetKAT \emph{normal forms for safety} (NFS). These normal forms are derived based on the additional equalities in Theorem~\ref{thm:neg-additional}.

\begin{theorem}[Distribution of $\neg$]\label{thm:neg-additional}
Let $a$, $b$ and $f = n_i$ for $i \in \{1, \ldots, m\}$ stand for NetKAT predicates as in Figure~\ref{fig::netKAT-syn-sem}. The following hold:
\[
\begin{array}{rcll}
\neg 1 & \equiv & 0 & \textnormal{NEG-ONE}\\ 
\neg 0 & \equiv & 1 & \textnormal{NEG-ZERO}\\
\neg (\neg a) & \equiv & a & \textnormal{NEG-NEG}\\
\neg (f = n_i) & \equiv & \Sigma_{j \not = i} f = n_j & \textnormal{NEG-ELIM}\\
\neg (a + b) & \equiv & (\neg a) \cdot (\neg b) & \textnormal{DIST-NEG-DISJ}\\
\neg (a \cdot b) & \equiv & (\neg a) + (\neg b) & \textnormal{DIST-NEG-CONJ}
\end{array}
\]
\end{theorem}
\begin{proof}[Proof Sketch]
All the above equivalences follow according to the NetKAT semantics in Figure~\ref{fig::netKAT-syn-sem}. Consider, for instance, {\small NEG-ONE}. The following holds:
\[
\begin{array}{rl}
\forall h \in {\textnormal{H}}: &\llbracket \neg 1\rrbracket h =_{(\it{def. of} \neg)}\\
& \{h\} \setminus (\llbracket 1 \rrbracket h) =_{(\it{def. of} 1)}\\
& \{h\} \setminus \{h\} = \\
& \{\} =_{(\it{def. of} 0)}\\
& \llbracket 0 \rrbracket h .
\end{array}
\]
\end{proof}

\begin{definition}[Token]\label{def::atom} We call a \emph{token} the identity policy $1$, the drop policy $0$, a test $(f = n)$, or a field modification $f \leftarrow n$.
\end{definition}

\begin{definition}[Normal Forms for Safety -- NFS]\label{def::NFS}
A NetKAT policy $p$ is in \emph{NFS} if 
\[
p \triangleq \Sigma_{i \in \{1, \ldots, m\}} ~\Pi_{j \in \{1, \ldots, n\}}tk_{i,j}
\]
with $tk_{i,j}$ a token, for all $i \in \{1, \ldots, m\}$ and $j \in \{1, \ldots, n\}$.
\end{definition}

\begin{theorem}[NFS reduction]\label{thm::NFS-red}
All policies defined over $\NetKATdup$ and repetitions can be reduced to equivalent policies in NFS.
\end{theorem}
\begin{proof}[Proof Sketch]
Let $p^u$ denote the repetition-free policy obtained from $p$ by performing all corresponding unfoldings, if any. It can be shown by induction on the structure of $p^u$ that an NFS can be obtained by applying the NetKAT axioms in Figure~\ref{fig::axioms-NetKAT}, together with the equalities in Theorem~\ref{thm:neg-additional} (in particular, {\small KA-SEQ-DIST-L} and {\small KA-SEQ-DIST-R}).
\end{proof}

\begin{definition}[$\sqsubseteq / \sqsubset$]\label{def::NFS-incl}
Let $p_i$ and $p'_j$ be NetKAT policies in NFS.
We write $p_i \sqsubseteq p'_j$ whenever $p_i$ can be obtained from $p'_j$ by deleting $k$ atoms at arbitrary positions in $p'_j$, with $k \geq 0$.
We write $p_i \sqsubset q_i$ whenever $k > 0$.

\end{definition}

\begin{definition}[Minimality]\label{def::min}
We call a policy in NFS \emph{minimal}, with
\[
p \triangleq \Sigma_{i \in \{1, \ldots, n\}} p_i
\]
whenever for all $p_j$ there is no $p_k$, with $j, k \in \{1, \ldots, n\}$ such that $p_j \sqsubset p_k$. 

Assume $p$ is in NFS, but is not minimal. We write $min(p)$ for the NFS policy obtained by removing all $p_k$, with $k \in \{1, \ldots, n\}$, such that there exists $p_j$, with $j \in \{1, \ldots, n\}$, satisfying $p_j \sqsubset p_k$.
\end{definition}

Assume an explanation ${\it expl} \not \equiv 0$ as in~(\ref{eq:star-elim-safe2}). Let ${\it expl}^{NFS}$ be {\it expl} reduced to its NFS.
The minimal explanation with respect to the violation of a safety property in NetKAT is represented by $min({\it expl}^{NFS})$.

\section{Tools for Explaining NetKAT Safety Failures}\label{sec::tool}
In this section we introduce {\tooleq}, a tool based on Maude~\cite{DBLP:conf/rta/ClavelDELMMT03},
for automatically computing relevant explanations for failures of NetKAT programs. Maude has been proven particularly suitable for defining semantics of programming languages and reasoning about their properties. The Maude tools encompass, amongst others, a suite of model checkers and the so-called Maude Formal Environment (MFE)~\cite{DBLP:conf/birthday/DuranRA11} which includes the Church-Rosser checker and the termination tool.
In short, {\tooleq} is based on Maude equational theories and it satisfies important properties such as Church-Rosser (which guarantees uniqueness of results) and termination. {\tooleq} provides all the explanations for NetKAT safety failures. 


\subsection{A Brief Overview of the Maude System}\label{sec::overview-Maude}
Maude specifications come in two flavours: (1) as functional modules, that define data types and associated operations by means of equational theories, or (2) as system modules, or rewrite theories, that specify concurrent transitions given as a set of rewrite rules, or ``oriented'' equations. Such rules are triggered whenever the rule’s left hand side matches a fragment of the system state and the rule’s condition is satisfied. In this work we utilize Maude functional modules and in the following we discuss the main aspects of Maude functional modules. We then continue with a brief overview of the MFE.

\medskip
{\it{Functional modules.}}
For an intuitive example, we next provide a Maude equational theory specifying NetKAT predicates.
First, note that a functional module is specified using the following syntax:
\begin{equation}\label{eq:fmod}
    \verb#fmod#~\emph{ModuleName}~\verb#is# ~\emph{DeclarationsAndStatements}~\verb#endfm#
\end{equation}

In our case, the module name is {\small{\verb#PREDICATE#}}, whereas the \emph{DeclarationsAndStatements} includes, amongst others, the operators defined according to the syntax in Figure~\ref{fig::netKAT-syn-sem}, and the associated axioms in Figure~\ref{fig::axioms-NetKAT}. Operators are specified over types, or Maude \emph{sorts}, defined within the current module via the keyword \verb#sort#, or imported (possibly in a ``protected'' fashion) from other modules. Properties such as associativity (\verb#assoc#), commutativity  (\verb#comm#), idempotency  (\verb#idem#), neutral elements  (\verb#id#) and precedence (\verb#prec#) can be specified as attributes of operators. Note that associativity and idempotency cannot be used together in any combination of attributes. Operators that play the role of constructors  (\verb#ctor#) for a certain type can also be specified; this is the case of all the operators defining Predicates in Figure~\ref{fig::netKAT-syn-sem}.
Variables (\verb#var#) of a certain sort can also be declared.
Possibly conditional equations are introduced using \verb#eq# or \verb#ceq#, respectively. Identifiers can be specified for equations as well. Comments are preceded by \verb#---#.

A Maude equational theory specifying NetKAT predicates and the additional boolean algebra axioms is given in Figure~\ref{fig:eqn-th-pred}.

The \emph{identity} and \emph{drop} NetKAT policies are defined in terms of two constants (or operators with arity $0$), namely, the constructors \verb#one# and \verb#zero#, respectively. \emph{Tests}, \emph{disjunction} and, respectively, \emph{conjunction} are straightforwardly implemented as the Maude binary operators \verb#_=_#, \verb#_+_# and, respectively, \verb#_._#.

Note that conjunction and disjunction are declared as associative and commutative as well. This is in accordance with the NetKAT axioms {\small KA-PLUS-ASSOC, KA-SEQ-ASSOC, KA-PLUS-COMM} and {\small BA-SEQ-COMM} in Figure~\ref{fig::axioms-NetKAT}. The advantage of using operator attributes is that Maude will efficiently perform equational reasoning modulo these attributes. \emph{Negation} is given as the unary operator \verb#~_#.
The remaining predicate axioms are specified via the equations in Figure~{{\verb#[BA-PLUS-ONE]#, \verb#[KA-PLUS-ONE]#, \verb#[KA-ONE-SEQ]#, \verb#[KA-ZERO-SEQ]#, \verb#[BA-EXCL-MID]#, \verb#[BA-CONTRA]#}} and {{\verb#[BA-SEQ-IDEM]#}}. Note that {\small KA-SEQ-ONE} and {\small KA-SEQ-ZERO} in Figure~\ref{fig::axioms-NetKAT} hold implicitly, due to the commutativity of sequential composition of NetKAT predicates.

Fields and their (natural) values are data structures defined within the corresponding Maude functional modules {\small{\verb#FIELD#}} and  {\small{\verb#NATVAL#}}, which {\small{\verb#PREDICATE#}} is importing in a protected manner.

\begin{figure}
\centering
\small{
\begin{verbatim}
(fmod PREDICATE is	
protecting FIELD .
protecting NATVAL .

sort Predicate .
var A : Predicate .

op one : -> Predicate [ctor] . 
op zero : -> Predicate [ctor] .	

op _=_ : Field NatVal -> Predicate [ctor prec 39] .
op _+_ : Predicate Predicate -> Predicate
                        [ctor assoc comm prec 43] .
op _._ : Predicate Predicate -> Predicate
                        [ctor assoc comm prec 40] .
op ~_ : Predicate -> Predicate [ctor prec 39] .

eq [BA-PLUS-ONE] : A + one = one .
eq [KA-PLUS-ZERO] : A + zero = A .
eq [KA-ONE-SEQ] : one . A = A .
eq [KA-ZERO-SEQ] : zero . A = zero .
eq [BA-EXCL-MID] : A + ~ A = one . 
eq [BA-CONTRA] : A . ~ A = zero .
eq [BA-SEQ-IDEM] : A . A = A .

eq ~ one = zero .
eq ~ zero = one .
endfm)
\end{verbatim}
}
\caption{Equational Theory of NetKAT Predicates.}
\label{fig:eqn-th-pred}
\end{figure}

\medskip
{\it{The MFE.}}
In our approach, we are using: Maude $2.7.1$ for Linux64\footnote{\url{http://maude.cs.illinois.edu/w/index.php/All_Maude_2_versions}}, MFE $1.0$b\footnote{\url{https://github.com/maude-team/MFE/wiki/How-to-use-the-tool}} including the Church-Rosser Checker (CRC) $3$p, and the Maude Termination Tool (MTT) $1.5$j, and AProVE~\cite{DBLP:journals/jar/GieslABEFFHOPSS17}. 

CRC plays a crucial role in resolving possibly different evaluations of a certain term by suggesting a series of so-called critical pairs. Intuitively, the latter are lemmas which, if proven correct, lead to a confluent equational specification. For instance, \verb#PREDICATE# is Church-Rosser because the following lemmas were soundly added to the specification of NetKAT predicates in Figure~\ref{fig:eqn-th-pred}, according to the additional equalities in Theorem~\ref{thm:neg-additional}:
\[
\begin{array}{cc}
  \verb#eq ~ one = zero .#   &~~~~~  \verb#eq ~ zero = one .#
\end{array}
\]

\subsection{Immediate Challenges and Observations}\label{seq:imm-chal}
In Figure~\ref{fig:eqn-th-pred} we presented a straightforward implementation of NetKAT predicates in Maude.
Next, we wanted to follow a similar approach and devise a Maude equational specification of NetKAT programs $\inn \cdot (1 + p \cdot t)^n \cdot \out$ as in~(\ref{eq:star-elim-safe2}). Recall that such programs are expressions defined over $\NetKATdup$ and repetitions $(-)^n$.

Typically, specifying such NetKAT policies would consist in the following straightforward steps:
\begin{enumerate}
\item Define a new sort {\small\verb#Policy#} as a suprasort of {\small\verb#Predicate#}.
\item Lift the signatures of $+$ and $\cdot$ to {\small\verb#Policy#}. 
\item Define $\leftarrow$ and the repetition operator $(-)^n$ accordingly.
\item Add the relevant set of axioms in Figure~\ref{fig::axioms-NetKAT} as Maude equations. (Recall that our approach for explaining safety failures discards the axioms for $*$, {\bf dup}, {\small PA-MOD-MOD} and {\small PA-FILTER-MOD}.) 
\end{enumerate}

Unfortunately, the recipe above was not successful.
We proceed by describing the main difficulties we encountered. 

\medskip
{\it Commutativity of $\cdot$~.}
Note that, on the one hand, the NetKAT $\cdot$ operator plays the role of conjunction in the context of predicates and is, therefore, commutative. On the other hand, $\cdot$ in the context of policies denotes sequential composition, which is not commutative. Nevertheless, the packet algebra axioms in Figure~\ref{fig::axioms-NetKAT} use $\cdot$ in a uniform fashion, thus, implicitly lifting $\cdot$ to the setting of policies as in step $2$ above. Consequently, defining in Maude two operators capturing the two different semantics of $\cdot$, and straightforwardly translating the axioms in Figure~\ref{fig::axioms-NetKAT} into equation is not an option.

\medskip
{\it Negation.}
The CRC returned a large number of critical pairs that involved the negation operator. Some of the pairs indicated the necessity of distributing negation over disjunction and conjunction as in Theorem~\ref{thm:neg-additional}. In accordance, we considered:
\begin{equation}\label{eq::dist-neg}
\begin{array}{rcll}
\neg (a + b) & \equiv & (\neg a) \cdot (\neg b) & \textnormal{DIST-NEG-DISJ}\\
\neg (a \cdot b) & \equiv & (\neg a) + (\neg b) & \textnormal{DIST-NEG-CONJ}\\
\end{array}
\end{equation}

Nevertheless, this did not help us eliminate all critical pairs either. Hence, we decided to apply a preprocessing step that reduces arbitrary NetKAT policies to equivalent negation-free policies in two steps. First, negations are pushed to the level of NetKAT predicates $f = n_i$ according to~(\ref{eq::dist-neg}). Then, each negated predicate $\neg (f = n_i)$ is soundly replaced according to:
\begin{equation}\label{eq::elim-neg}
\neg (f = n_i) \equiv \Sigma_{j \not = i} f = n_j ~~~~~~ \textnormal{NEG-ELIM}   
\end{equation}
As in~\cite{DBLP:conf/popl/AndersonFGJKSW14}, field values are drawn from finite domains.

\medskip
{\it Distributivity.} We also noticed that the distributivity axioms {\small BA-PLUS-DIST}, {\small KA-SEQ-DIST-L} and {\small KA-SEQ-DIST-R} contribute to the violation of the Church-Rosser property when used together within the equational theory of policies. For instance, 
\[
(a + b) \cdot (a + c)
\]
can be reduced according to {\small BA-PLUS-DIST} to: 
\begin{equation}\label{eq:ba-plus-dist}
a + b \cdot c
\end{equation}
and it can be reduced according to {\small KA-SEQ-DIST-R} and {\small BA-SEQ-IDEM}, to:
\begin{equation}\label{eq:ba-seq-dist}
a + b\cdot a + a\cdot c + b\cdot c.
\end{equation}
From the perspective of safety failure explanations, the policy in~(\ref{eq:ba-seq-dist}) subsumes its counterpart in~(\ref{eq:ba-plus-dist}). Hence, {\small BA-PLUS-DIST} can be discarded as well.

\subsection{Equational Specifications for Explaining Failures}\label{seq::tooleq}
In this section we introduce {\tooleq}, a tool for explaining NetKAT safety failures. {\tooleq} is based on the Maude equational specification $\NetKATdup$, implemented in a manner that enables accommodating the ideas in Section~\ref{seq:imm-chal}. The functional modules behind {\tooleq} are proven Church-Rosser and terminating. Hence, {\tooleq} provides the unique solution encoding all relevant explanations on how packets can travel from a specified ingress to the undesired egress. 

Assume the $\NetKATdup$ policies encoding a network topology $t$, a switch policy $p$, an ingress policy $\inn$, and an egress policy $\out$ encoding an undesired property. Let $P \triangleq \inn \cdot (1 + p \cdot t)^n \cdot out$ be the corresponding NetKAT program to be analyzed for safety failures.
{\tooleq} works in three steps.

(I) Firstly, the tool recursively unfolds the policy $(1 + p \cdot t)^n$ into a term $U$. Then, $U$ is reduced to a term $F$ uniquely expressed as a sum of policies that are union-free and negation-free. This is achieved in accordance with the equivalences~(\ref{eq::dist-neg}) and~(\ref{eq::elim-neg}) in Section~\ref{seq:imm-chal}, and with the distributivity axioms {\small KA-SEQ-DIST-L} and {\small KA-SEQ-DIST-R}, respectively.

(II) Next, $F$ is reduced to $F'$ according to the relevant NetKAT axioms implemented in Maude in a slightly modified fashion, due to the issues related to the commutativity of $\cdot$, as discussed in Section~\ref{seq:imm-chal}.

For an intuition, consider a (possibly conditional) NetKAT axiom generically denoted by $l \cdot r \equiv t ~(\textnormal{if} ~C)$. With a commutative $\cdot$, it might be the case that $F$ can be equivalently represented as a term $F'$ within which $l \cdot r$ can be matched (whenever $C$ holds). Nevertheless, given that a commutative $\cdot$ could not be considered in the Maude specification of NetKAT policies, it might be the case that $l \cdot r$ does not match in $F'$ (even if $C$ holds).
Consequently, the aforementioned axiom might not be employed by the Maude equational reduction procedure, when starting with $F'$.

The solution is to enable sound reductions according to $l \cdot r \equiv t ~(\textnormal{if} ~C)$, in all possible contexts. More precisely, each such axiom is implemented via a set of equations of shape:

\[
\begin{array}{rcl}
l \cdot r  & \equiv & t ~(\textnormal{if} ~C)\\
l \cdot M \cdot r & \equiv & t ~(\textnormal{if } C \textnormal{ and}~ C_s)\\ 
\end{array}
\]
where $M$ is a policy term and $C_s$ is a condition that ensures the sound application of the newly introduced equations. For an example, we next provide a corresponding Maude implementation of the {\small PA-CONTRA}.

{\fontsize{9}{12}\selectfont
\begin{verbatim}
ceq (F1 = I1) . (F1 = I2) = zero if I1 =/= I2 .
ceq (F1 = I1) . M . (F1 = I2) = zero if I1 =/= I2 /\ not (F1 <- I2 occursInner M)  .
\end{verbatim}
}

Intuitively, \verb#(F1 <- I2 occursInner M)# checks whether the field modification \spverb#F1 <- I2# occurs within the policy \verb#M#. \spverb#(F1 <- I2 occursInner M)# is evaluated to \verb#true# whenever the field modification \verb#F1 <- I2# occurs within \verb#M#. Otherwise, \verb#(F1 <- I2 occursInner M)# is evaluated to \verb#false#. We negate the result obtained from performing this check and this way, the second equation soundly equates its left-hand side to \verb#zero#, as the field \verb#F1# is never modified with the value \verb#I2# within \verb#M# and the initial value of the field \verb#F1# is different than \verb#I2#, hence the test \verb#F1 = I2# will always fail.

We then apply certain axioms in order to simplify the expressions. For an example, we provide the implementation of {\small BA-SEQ-IDEM} axiom.

{\fontsize{9}{12}\selectfont
\begin{verbatim}
eq A . A = A .
ceq (F1 = I1) . M . (F1 = I1) = (F1 = I1) . M if M ? F1 .
\end{verbatim}
}

\noindent where \verb#A# is of sort predicate.  The  operator \verb#?# works in a similar fashion to the operator \verb#occursInner#. Intuitively, \verb#occursInner# checks whether a specific term occurs inside a given policy, whereas the operator \verb#?# only checks whether there exist an assignment to a field in a given policy. The term \verb#M ? F1# is evaluated to \verb#true# whenever \verb#F1# is not modified within \verb#M#. Otherwise, \verb#M ? F1# is evaluated to \verb#false#. This way, it is ensured that the term \verb#F1 = I1# can commute inside the terms in \verb#M# as \verb#F1# is not modified within \verb#M#, and then {\small BA-SEQ-IDEM} axiom can be applied. 

Another phase in this step is to define a total order between the fields and reorder the terms according to this total order. This phase is needed to obtain canonical forms. We introduce the operator \verb#<# to define the total order and we then apply the following equations to bring the expressions into a canonical form.

{\fontsize{9}{12}\selectfont
\begin{verbatim}
ceq (F1 <- I1) . (F2 <- I2) = (F2 <- I2) . (F1 <- I1) if F1 < F2 .
ceq (F1 = I1) . (F2 = I2) = (F2 = I2) . (F1 = I1) if F1 < F2 . 
\end{verbatim}
}

(III) Last, but not least, if the reduction at step (II) returns the unique term $F'' \not \equiv 0$ encoding all safety failure explanations,
then {\tooleq} computes all relevant explanations when starting with $F''$, according to the minimization procedure in Section~\ref{sec:minimality}.

The full implementation of {\tooleq} can be downloaded at:\\
{\small
\url{https://gitlab.inf.uni-konstanz.de/huenkar.tunc/sdn-safecheck}.
}

\section{Experimental Evaluation}\label{sec:experiments}
We performed experiments to evaluate the performance of our implementation on the publicly available Topology Zoo dataset~\cite{DBLP:conf/pam/GillALM08} which consist of $261$ real-world network topologies. Given that, in essence, safety failure analysis reduces to reachability analysis, in our experiments we analyzed the time required to check for reachability within these topologies. More precisely, we checked point-to-point reachability between the two nodes in the longest path within the network. If there were more than one such paths, then an arbitrary choice was made. We encoded the topologies in the dataset into NetKAT and generated a destination-based shortest path policy to connect each node with every other node by using an automated procedure similar to the one in~\cite{DBLP:conf/pldi/BeckettGW16}. The encoded topologies are made available in the link above alongside the implementation of the tool. All the experiments were performed on a computer running Ubuntu 18.04 LTS with 8 core 3.7GHz AMD Ryzen 7 2700x processors and 32 GB RAM.

\begin{figure}[ht] 
\center
\includegraphics[scale=0.109,origin=c]{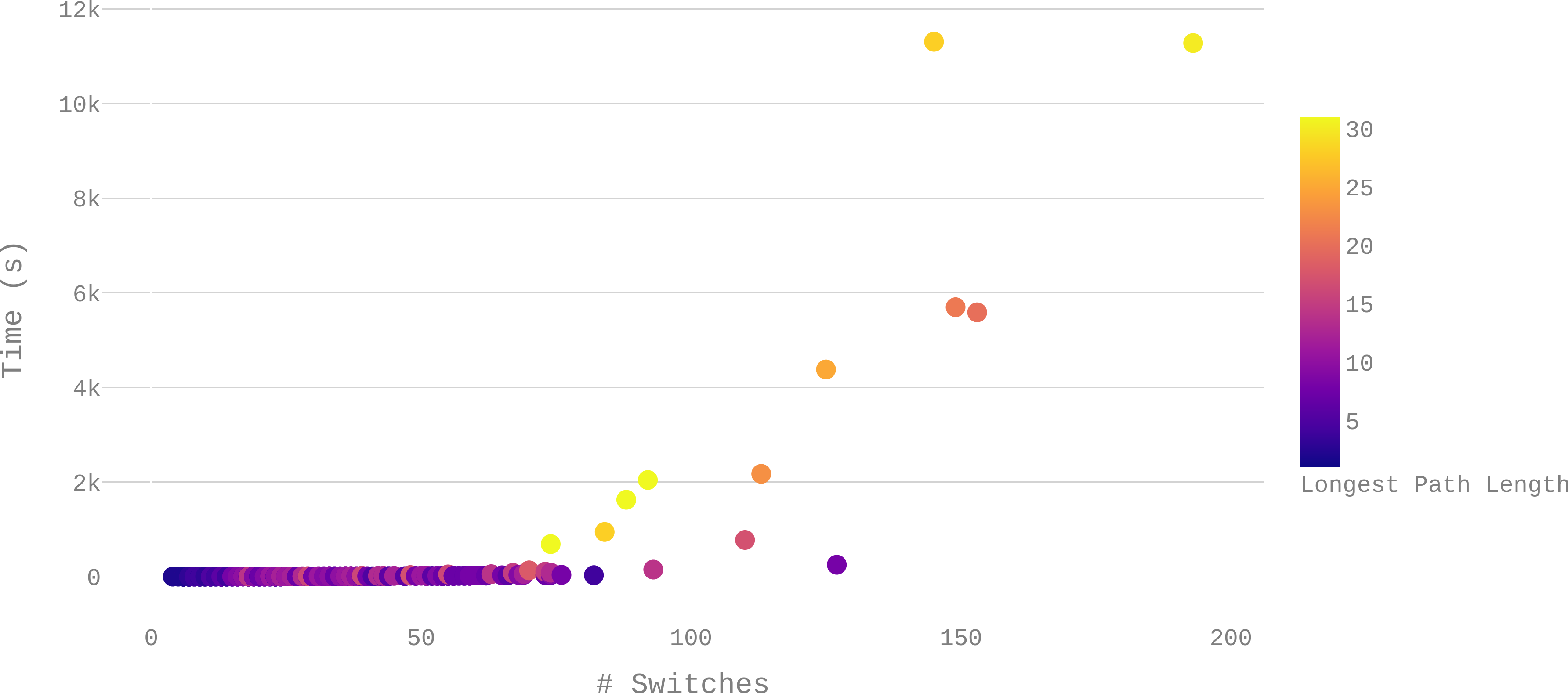}
\caption{Experimental results \label{fig::result}} 
\end{figure}

A scatter plot of the obtained execution times is sketched in Figure~\ref{fig::result}. We set a time limit of $12000$ seconds for checking the reachability property. For three topologies the computation did not finish under this time limit. The networks for which the computation timed out consist of $754$, $197$ and $153$ nodes, and correspond to first, second and fourth largest network in the Topology Zoo dataset, respectively. The results show that for networks up to $70$ switches a result is obtained under $60$ seconds in most cases. For networks with more than $70$ switches the variance of the obtained execution times is higher. We observe that the longest path length plays a significant role in determining the running time of {\tooleq} as networks grow in size. 

\begin{figure}[ht] 
\center
\includegraphics[scale=0.108,origin=c]{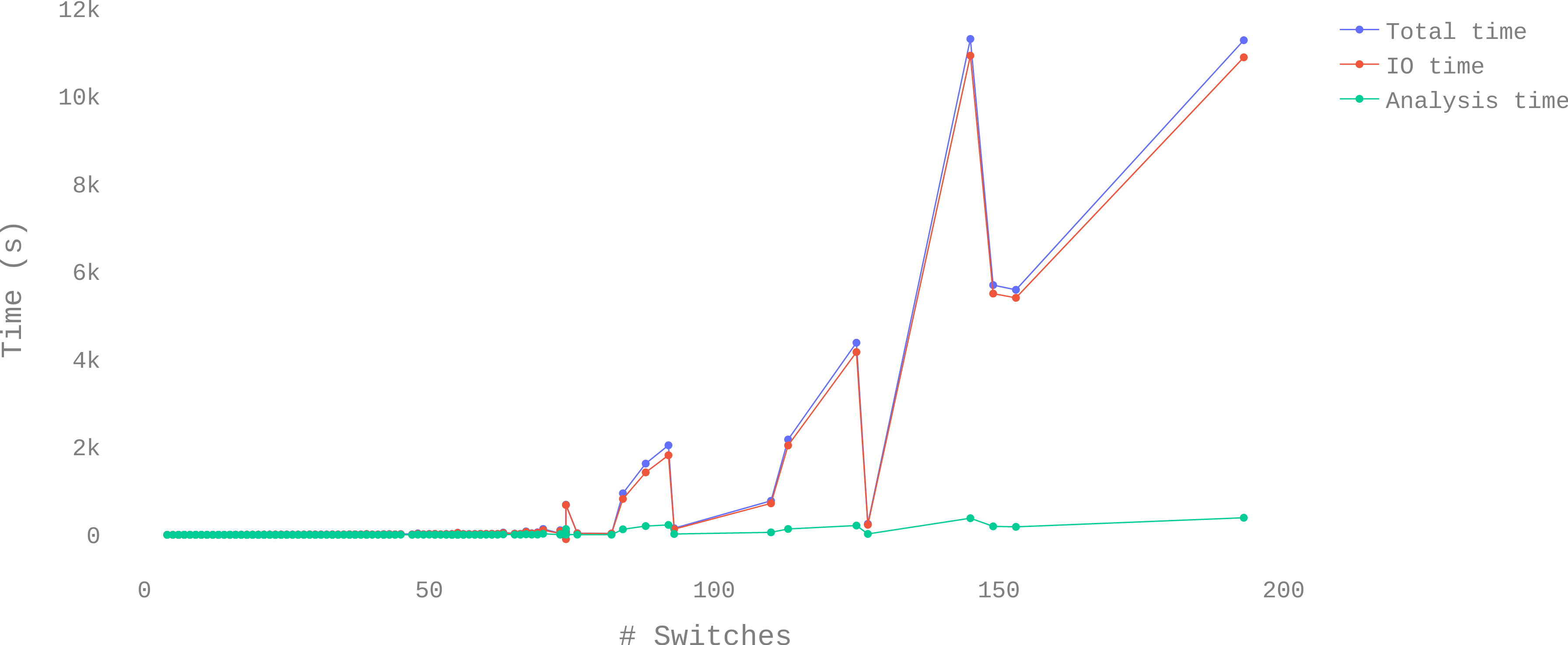}
\caption{Time comparisons \label{fig::time-comparison}} 
\end{figure}

The execution time can be divided into two categories: IO time and analysis time. The IO time corresponds to the time frame in which the expressions are written into a file and loaded into Maude. Analysis time corresponds to the time frame in which the rewriting and the failure analysis is performed. In Figure~\ref{fig::time-comparison} we display a comparison between the time taken for IO and the time taken for performing the analysis. We observe that the IO time dominates the total execution time. 

\section{Conclusions}\label{sec::discussion}

In this paper we formulate a notion of safety in the context of NetKAT programs~\cite{DBLP:conf/popl/AndersonFGJKSW14} and provide an equational framework that computes all relevant explanations witnessing a bad, or an unsafe behaviour, whenever the case. The proposed equational framework is a slight modification of the sound and complete axiomatisation of NetKAT and, as shown by the experimental evaluation, is parametric on the size of the underlying network topology. The new equational system is not complete, as some of the original NetKAT axioms have been removed to enable more comprehensive failure explanations. Nevertheless, the purpose of our framework is not to reason about equivalence, but to identify safety failure violations and corresponding explanations.

Our approach is orthogonal to related works which rely on model-checking algorithms for computing all counterexamples witnessing the violation of a certain property, such as~\cite{DBLP:conf/vmcai/Leitner-FischerL13,DBLP:journals/corr/abs-1901-00588}, for instance.
The Maude system was exploited for implementing {\tooleq} tool for automatically computing safety failure explanations.
Corresponding experimental evaluation based on the Topology Zoo dataset~\cite{DBLP:conf/pam/GillALM08} is also provided.

The results in this paper are part of a larger project on (counterfactual) causal reasoning on NetKAT.
In~\cite{Lew73}, Lewis formulates the counterfactual
argument, which defines when an event is considered a cause for some
effect (or hazardous situation) in the following way: a) whenever the event
presumed to be a cause occurs, the effect occurs as well, and b) when the presumed cause does not occur,
the effect will not occur either.
The current result corresponds to item a) in Lewis' definition, as it describes the events that have to happen in order for the hazardous situation to happen as well. The next natural step is to capture the counterfactual test in b). This reduces to tracing back the explanations to the level of the switch policy, and rewrite the latter so that it disables the generation of the paths leading to the undesired egress. The generation of a ``correct'' switch policy can be seen as an instance of program repair.

In the future we would be, of course, interested in defining notions of causality (and associated algorithms) with respect to the violation of other relevant properties such as liveness, for instance. We would also like to explain and eventually disable routing loops (i.e., endlessly looping between A and B) from occurring. Or, we would like to identify the cause of packets being not correctly filtered by a certain policy.

\paragraph{Acknowledgements}{The authors are grateful to Francisco Dur\' an, Steven Eker and the Maude/RL community for their useful comments on using the Maude Formal Environment,
and to the reviewers of FROM 2019, for their feedback and observations.
Special thanks are addressed to Marcello Bonsangue and Tobias Kapp\'e, for their insight into the formal foundations of NetKAT. Many thanks to Hossein Hojjat and Dang Mai for their insight into the behaviour of SDNs and associated programming languages.
This work was supported by the DFG project ``CRENKAT'', proj. no. $398056821$.}

\bibliographystyle{elsarticle-num} 
\bibliography{generic}

%
%
%
\end{document}